\documentclass[a4paper]{amsart}
\usepackage{fullpage}
\usepackage{dsfont,amsmath,amsthm,amssymb}

\usepackage{subfigure}

\usepackage{tikz}
\usetikzlibrary{shapes.symbols,arrows,decorations.pathmorphing}
\usetikzlibrary{matrix}

\usepackage{enumerate}
\usepackage{fullpage}

\newtheoremstyle{newthm}
  {\topsep}   
  {\topsep}   
  {\itshape}  
  {0pt}       
  {\scshape} 
  {.}         
  {5pt}  
  {}          
\DeclareMathOperator{\avg}{avg}
\theoremstyle{newthm}

\newtheorem{thm}{Theorem}
\newtheorem{prop}[thm]{Proposition}
\newtheorem{lem}[thm]{Lemma}
\newtheorem{cor}[thm]{Corollary}

\newcommand{\IR}{\mathds{R}}
\newcommand{\IN}{\mathds{N}}
\newcommand{\IZ}{\mathds{Z}}

\newcommand{\SCC}{\mathds{T}}

\renewcommand{\le}{\leqslant}
\renewcommand{\leq}{\leqslant}
\renewcommand{\ge}{\geqslant}
\renewcommand{\geq}{\geqslant}

\newcommand{\xb}{{\overline{x}}}
\newcommand{\xib}{{\overline{x}_i}}
\newcommand{\xjb}{{\overline{x}_j}}
\newcommand{\xkb}{{\overline{x}_k}}
\newcommand{\yb}{{\overline{y}}}
\newcommand{\zb}{{\overline{z}}}
\newcommand{\vect}{{\mathrm{vect}}}

\usepackage{dsfont,mathtools}

\newtheorem{claim}[thm]{Claim}

\title{A Proof of the Convergence of the Hegselmann-Krause Dynamics on the Circle}

\author{Bernadette Charron-Bost}
\address[Bernadette Charron-Bost]{CNRS, \'Ecole polytechnique}
\email{charron@lix.polytechnique.fr}

\author{Matthias F\"ugger}
\address[Matthias F\"ugger]{MPI Saarbr\"ucken}
\email{mfuegger@mpi-inf.mpg.de}

\author{Thomas Nowak}
\address[Thomas Nowak]{ENS Paris}
\email{thomas.nowak@ens.fr}

\begin{document}

\maketitle

\section{Introduction}

In this note, we give a complete proof that Hegselmann-Krause systems  converge on the circle following  
	the proof strategy developed by Hegarty, Martinsson, and Wedin  in~\cite{HMW14, HMW15}.
By letting $$K_s = \sum_{t\in \IN} \sum_{i\in[n]} \delta(\xb_i(t) ,  \xb_i(t+1)
)^s$$ the $s$-kinetic energy of the system,
	their proof strategy consists in showing that (1) the quadratic kinetic energy $K_2$ is finite,
	(2) the influence graph is eventually constant, and (3) the 1-kinetic energy $K_1$ is finite, which
	immediately implies  the convergence of the sequence of position vectors $\left(  \xb(t) \right)_{t \in \IN}$.
	
To show the finiteness of $K_2$, we present a simple  proof in Section~\ref{sec:lyapunov}
	which is based on a reduction of the HK dynamics on the circle 
	to the HK dynamics on the line.
	
Concerning the eventual stability of influence graphs, we are not able to understand the proof outlined in~\cite{HMW14},
	and we give our proof of this key point in Section~\ref{sec:graph}.
	
For the third point, namely the finiteness of $K_1$, Hegarty, Martinsson, and Wedin introduce the  vector of 
	position differences $x^*(t)$, and show that the sequence  $\left( x^*(t) \right)_{t\in \IN}$ converges to 
	some limit $x^*_{\infty}$ such that $\lVert  x^*(t) - x^*_{\infty} \lVert = O\left( e^{-ct} \right)$.
The argument given in the first version~\cite{HMW14} for the latter point is erroneous, and we fix the argument  in Section~\ref{sec:gradient}.
An alternative argument, based on the finiteness of the quadratic kinetic energy,  is provided in~\cite{HMW15} which 
	still appeals  to the vector of differences $x^*(t)$.
In fact, this argument can be directly applied to the position vector as we
show in Section~\ref{sec:v2}, which allows to circumvent the arguments of
Sections~\ref{sec:gradient} and~\ref{sec:conv} and prove convergence directly.

\subsection*{Notation}
\subsubsection*{The circle}
Let $p $ be a positive real number and let $\SCC = \IR / p\IZ$.
We define $\delta$ on $\SCC$ by:
	$$\delta(\xb, \yb) = \min_{x\in\xb, y\in \yb}|x-y|  \, ,$$
	and we easily check that $\delta$ is a distance on $\SCC$ such that 
	$$0 \leq \delta(\xb, \yb) \leq p/2 \, .$$
Let $\vect(\xb,\yb)$ be the unique element in $\yb - \xb$ which lies in $]-p/2,p/2]$.
Clearly $$| \vect(\xb,\yb) | = \delta(\xb, \yb)\, .$$
When $\xb, \yb, \zb$ are close enough, we have
	$$ \vect(\xb,\zb) =  \vect(\xb,\yb) +  \vect(\yb,\zb) \, ,$$
	but the equality does not hold in general.
	
\noindent This induces a bijection  $ \phi \ : \ \SCC \rightarrow ] -p/2,p/2]$.
	
\subsubsection*{The HK dynamics on the circle}

The Hegselmann-Krause dynamics on the circle with an influence radius $r \in [0,p/2]$  is defined by:
	\begin{equation}\label{eq:hk}
	 \vect(\xib(t), \xib(t + 1)) = \frac{1}{| N_i(t)|} \sum_{j \in N_i(t)}  \vect(\xib(t), \xjb(t)) 
	 \end{equation}
	where $$ N_i(t) = \{ j \in [n] \ | \   \vect(\xib(t), \xjb(t)) \in [-r , r] \} \ .$$
Indeed we have $$-p/2 < \frac{1}{| N_i(t)|} \sum_{j \in N_i(t)}  \vect(\xib(t), \xjb(t)) \leq p/2 \, , $$
	and (\ref{eq:hk}) defines a unique element $\xib(t+1)$ in $\SCC$.
The sets $N_i(t)$ define a directed graph, called  the {\em influence graph at time $t$} and
	denoted by $G_t$.

At each time $t$ and for each agent~$i$, we define the set of left neighbors $L_i(t)$
	and  the set of  right neighbors $R_i(t)$:
	$$L_i(t) = \{ j \in [n] \ : \  \vect(\xjb(t), \xib(t)) \in [0 , r] \} \ \mbox{ and } 
	R_i(t) = \{ j \in [n] \ : \  \vect(\xib(t), \xjb(t)) \in [0 , r] \} \ .$$
For each agent~$i$, let $r_i(t)$ be the (unique)  integer at most equal to $i$ and such that 
	$$ |R_i(t)|  \equiv  r_i(t) -i + 1  \mod n \, .$$
Similarly we define $\ell_i(t)$ for the set $L_i(t)$.

The mapping $\phi$ naturally induces an ordering on the initial positions, and
	we assume that 
	$$\phi(  \overline{x}_{1}(0) ) \leq \dots \leq \phi(\overline{x}_{n}(0)) \, .$$
Like for the line, the HK dynamics on the circle cannot cause agents to cross 
	(even if there may be some cyclic shifts in the ordering with respect to $\phi$).
We say that {\em there is a cut at time $t$ } if there are two consecutive agents $i$ and $i+1$ 
	such that $i +1 \notin R_i(t)$.
We easily observe that if there is a cut at time $t$, the system remains cut forever and
	the dynamics is similar to the HK dynamics on the line. 

In the rest of this note, we fix the influence radius $r=1$.

\section{Lyapunov Function}\label{sec:lyapunov}

We consider an HK system on the circle with $n$ agents and an $n$-vector of initial positions.
From this system, we construct an HK system on the line with $nN$	agents by ``unrolling''  $N$
	samples of the HK system on the circle, and we denote the position of agent~$i$ at time $t$
	for this system by $y_i(t)$.
Then we introduce
	$$ V_N(t) = \sum_{(i,j)\in [nN]^2} \min \big( 1, |y_i(t) - y_j(t)|^2 \big) $$ 
	and $$ W(t) = \sum_{(i,j)\in [n]^2} \min \big( 1,  \delta( \xib(t) , \xjb(t) )^2 \big) $$ 
	for the HK systems on the line and the circle, respectively.

If $N\geq 4$, then we easily check that 
	$$\begin{array}{ll}
	V_N(0)   =  (N-2) W(0) + R_0 & \mbox{ with } 0\leq R_0 \leq 2n^2 \\
	V_N(1)   =  (N-4) W(1) + R_1 & \mbox{ with } 0\leq R_1 \leq 4n^2 \ .
	\end{array}$$
Moreover by definition of $W$, we have $$ 0 \leq W(0) \leq n^2 \ .$$
From~\cite{RMF08} we know that $$ V_N(0) - V_N(1) \geq 4 \sum_{i \in [nN]} |y_i(1) - y_i(0) |^2    \geq 4 \sum_{i =n+1}^{(N-1)n} |y_i(1) - y_i(0) |^2   \ .$$
We let $$ S =  \sum_{i \in [n]} \delta(\xib(0) , \xib(1))^2 \ .$$
Since $R_0\leq 2n^2$, $0 \leq R_1$, and $W(0) \leq n^2$, for every integer $N\geq 4$ we have
	$$W(0) - W(1) \geq 4 \left ( \frac{N-2}{N-4} S - \frac{n^2}{N-4} \right ) \ .$$
When $N$ tends to $+\infty$, that gives $$ W(0) - W(1) \geq 4 S \ . $$
Because the above inequality holds whatever the initial positions of the agents on the circle, 
	we derive the following proposition.
	
\begin{prop}\label{prop:W}
At each time $t\in \IN$,
		$$ W(t) - W(t+1) \geq 4 \sum_{i \in [n]} \delta(\xib(t) , \xib(t+1))^2 \ .$$
\end{prop}

Therefore the sequence $\left( W(t)\right)_{t\in\IN}$ is decreasing and nonnegative, and so converges
	to some $W(\infty)$.
Following~\cite{Cha11}, we define the $s$-{\em kinetic energy} of the HK system on the circle by
	$$K_s =  \sum_{t \in \IN} \sum_{i \in [n]} \delta(\xib(t) , \xib(t+1))^s \  $$
	where $s$ is a real number.
Obviously the finiteness of $K_1$ enforces the convergence of the sequence  $\left(  \xb(t) \right)_{t \in \IN}$.

Since at each time $t$, we have $0 \leq W(t) \leq n^2$, we derive the following theorem 
	from Proposition~\ref{prop:W}.

\begin{thm}\label{thm:kinetic2}
The 2-kinetic energy of an HK system on the circle with $n$ agents satisfies
	$$ K_2 \leq n^2/4 \, .$$
\end{thm}

\vspace{0.2cm}
Unfortunately the finiteness of $K_2$ is not sufficient to enforce the convergence of the sequence  
	$\left(  \xb(t) \right)_{t \in \IN}$.
The  proof below  consists in showing that for the HK dynamics on the circle, 
	the 1-kinetic energy $K_1$ is also finite, which does imply the convergence of  
	$\left(  \xb(t) \right)_{t \in \IN}$.

\section{Topological changes}\label{sec:graph}

In this section, we study the impact on the kinetic energy of changes in the influence graph.
	
Suppose that the influence graph $G_{t+1}$ at time $t+1$ contains a link that
	is not a link at time~$t$, in $G_t$. 
In the case of HK on the line,  the agent with a new left neighbor 
	that has the greatest identity has no new right neighbor since
	influence graphs are bidirectional.
Actually we show that this key point for the study of the HK dynamics on the line also holds on the circle.

\vspace{0.2cm}
\begin{prop}\label{prop:simuline}
If  $G_{t+1}$ contains a link that is not a link in $G_t$, then there is an agent
	that has a new left neighbor but no new right neighbor at time  $t+1$.
\end{prop}

\begin{proof}
We proceed by contradiction: suppose that such an agent does not exist.
Because of the above remark, the system is not cut at time $t+1$.
Let $i_1$ be the agent with the smallest identity and a new left neighbor;
	let $i_2$ be one of the new right neighbor of~$i_1$.
We repeat the construction and obtain an infinite chain of agents $i_1,\ i_2, \dots$

By the pigeonhole principle, this chain is closed, i.e., there exist two indices $k$ and $\ell$
	such that $i_k = i_{\ell}$.
That gives two  closed chains of elements in $\SCC$, namely  $x_{i_k}(t), \dots, x_{i_{\ell}}(t)$ 
	and $x_{i_k}(t + 1), \dots, x_{i_{\ell}}(t + 1)$,
	with the same length which is a multiple of $p$, say $\nu p$ since the system is not cut.
By construction, we have:
	$$ (k - \ell) r < \nu p \leq 	(k - \ell) r \ ,$$
	a contradiction.
\end{proof}

\begin{prop}\label{prop:addlink}
If agent~$i$ has a new left neighbor but  no new right neighbor at time $t+1$ in the influence
	graph, i.e.,  $ L_i(t) \subsetneq L_i(t+1) $ and $ R_i(t+1) \subseteq R_i(t)$,
	then there is at least one agent~$j$ which moves by more than  $1/6n$ at time $t+1$ or $t+2$
	$$  \delta( \xib(t) , \xib(t + 1) ) > \frac{1}{6n}   \   \mbox{ or }  \ \delta( \xib(t + 1) , \xib(t + 2) ) > \frac{1}{6n}  \, .$$
\end{prop}

\begin{proof}
We proceed by  contradiction and assume that
     no agent moves from time~$t$ to time~$t+2$ by more than~$\mu = \frac{1}{3n}$, i.e., for
     each agent $\ell \in [n]$,
     $$\delta ( \xb_\ell(t+1) , \xb_\ell(t+2) ) + \delta( \xb_\ell(t) , \xb_\ell(t+1) ) < \mu \, .$$

By hypothesis, we have
	$$ L_i(t) \subsetneq L_i(t+1)   \ \mbox{ and }   R_i(t+1) \subseteq R_i(t) \, .$$   
Without loss of generality we assume that~$\xb_i(t) = 0$, which allows us to
identify $\xb_j(t)$ with $\phi\big( \xb_j(t) \big)$ in the rest of this proof.
It is, 
\begin{align}
\xb_i(t+2) = \left(\sum_{k  \in L_i(t+1)}\xb_k(t+1) \right)/ \left(|N_i(t+1)| \right) \, . \label{eq:xtp2}
\end{align}
We will show that $\xb_i(t+2) < -\mu$, which  contradicts
   the assumption that no agent moves by more than~$\mu$ from time~$t$ to time~$t+2$
   since  $\xb_i(t) = 0$.

For every agent~$j \in L_i(t+1)\setminus L_i(t)$,  it holds that $\xb_j(t) < -1$ and
     thus $ \xb_j(t+1) < -1+\mu $.
Let $m = | L_i(t+1)\setminus L_i(t) | $.
We have 
\begin{align}
\xb_i(t+2)  \leq \mu - \frac{m}{|N_i(t+1)|} + \frac{\sum_{k  \in L_i(t) \cup R_i(t+1) }\xb_k(t) }{|N_i(t+1)|} \, .
\end{align}
Since $m \ge 1$ and $|N_i(t+1)|  \le n$, 
\begin{align}
\xb_i(t+2)  \leq \mu - \frac{1}{n} + \frac{\sum_{k  \in L_i(t) \cup R_i(t+1) }\xb_k(t) }{|N_i(t+1)|} \, .
\end{align}

\begin{claim}\label{lem:avg}
Denote by $\avg(X)$ the (equal weight) average of a multiset~$X$ with
     elements in $\IR$.
Let $A$ and $B$ be two multisets with elements in $\IR$.
If $A \subseteq B$ and all elements in multiset $B\setminus A$ are greater or equal
     than $\max(A)$, then $\avg(B) \ge \avg(A)$.
\end{claim}

We now apply the above claim with $A$ defined by the elements 
	$\xb_k(t)$ with $k\in L_i(t) \cup R_i(t+1)$ and 0 repeated $m = |L_{i}(t=1) \setminus L_i(t) |$ times, 
	and $B$  defined with the elements in $A$ and  the $\xb_k(t)$'s with $k\in R_i(t) \setminus R_i(t+1)$.
Thus we obtain
	$$ \frac{\sum_{k  \in L_i(t) \cup R_i(t+1) }\xb_k(t) }{|N_i(t+1)|} \leq \frac{\sum_{k  \in L_i(t) \cup R_i(t) }\xb_k(t) }{|L_i(t+1) \cup R_i(t) |} \, .$$
Therefore 
	$$ \frac{\sum_{k  \in L_i(t) \cup R_i(t+1) }\xb_k(t) }{|N_i(t+1)|} \leq \frac{ |N_i(t)| }{  m + |N_i(t)  |} \, \xib(t+1) \, .$$
It follows that 
	\begin{align}
	\xb_i(t+2)  - \xib(t+1) \leq \mu - \frac{1}{n} +  \frac{ |N_i(t)| }{  m + |N_i(t)  |} \, \xib(t+1) \leq  \mu - \frac{1}{n} +  \frac{ |N_i(t)| }{  m + |N_i(t)  | } \, \mu \, ,
\end{align}
	and so 
		\begin{align}
	\xb_i(t+2) - \xib(t+1)  <  2\mu- \frac{1}{n}  \, .
\end{align}
For	$\mu= \frac{1}{3n}$, we actually obtain $\xb_i(t+2) - \xib(t+1)  <  -\mu$, a contradiction.
\end{proof}

Combining Proposition~\ref{prop:addlink} and  Theorem~\ref{thm:kinetic2}, we derive that 
	the number of times a link is added, and so the number of changes in the influence graph, is finite.
\vspace{0.2cm}
\begin{thm}\label{thm:finitegraphchanges}
In the HK dynamics on the circle, the influence graph is eventually constant.
\end{thm}
\vspace{0.2cm}
That leads us to decompose the HK dynamics into two periods: the first one during which
	the influence graph changes and the second one with a fixed influence graph.
Theorem~\ref{thm:finitegraphchanges} shows that the first period is finite while the second
	one may be infinite as exemplified in~\cite{HMW14}.
Observe that if there is a cut, then the dynamics on the circle is the same as on the line,
	in which case the second phase is of length one and the system freezes
	just after the first period. 

\section{The vector of position differences }\label{sec:gradient}

We define $$x_i^*(t) = \left\{ \begin{array}{ll}
                                                    \vect(\xib(t) , \overline{x}_{i+1}(t)) & \mbox{ if } i\in [n-1] \\
					       \vect( \overline{x}_{n}(t) , \overline{x}_{1}(t) ) & \mbox{ if } i = n \, . 
					       	\end{array} \right. $$
If there is no cut at time $t$, then we have the fundamental identity
	\begin{equation}\label{eq:perimeter}
	x_1^*(t) + \dots + x_n^*(t) = p \, .
	\end{equation}
						
\begin{prop}\label{prop:nonnegative}
At each time $t$, there exists a non-negative matrix $A_t$ such that $$x^*(t+1) = A_t \, x^*(t) \, .$$
\end{prop}						

\begin{proof}
We fix time $t$ and we often omit $t$ in the notation as no confusion can arise.

First we observe that under the condition $r < p/6$, we have
	$$ x_i^*(t+1) = x_i^*(t) + \frac{1}{|N_{i+1}(t)|} \sum_{k\in N_{i+1}(t)} \vect(\overline{x}_{i+1}(t) , \xkb(t)) 
	- \frac{1}{|N_i(t)|} \sum_{k\in N_{i}(t)} \vect(\overline{x}_{i}(t) , \xkb(t))\, .$$
From this expression of $ x_i^*(t+1)$, we derive that $x_i^*(t+1) = \sum_{k \in [n]} A_{i \, k}  \, x_k^*(t) $ with

\begin{equation}\label{eq:entries}
	A_{i \, k} = \left\{\begin{array}{ll}
		\frac{k-\ell_i +1}{r_i - \ell_i +1} & \mbox{ if } \ell_i \leq k \leq \ell_{i+1} - 1 \\ \\
		\frac{k-\ell_i +1}{r_i - \ell_i +1} - \frac{k-\ell_{i+1} +1}{r_{i+1} - \ell_{i+1} +1} & \mbox{ if } \ell_{i + 1}  \leq k \leq i - 1 \\ \\
		\frac{r_{i+1} -k}{r_{i+1} - \ell_{i+1} +1} - \frac{r_i -k }{r_i - \ell_i +1}  & \mbox{ if } i  \leq k \leq r_{i }- 1 \\ \\
		\frac{r_{i+1} -k}{r_{i+1} - \ell_{i+1} +1} & \mbox{ if } r_i \leq k \leq r_{i+1} - 1 \\ \\ 
		0 & \mbox{ else} \, .
		\end{array}
		\right.  \end{equation}
Then we prove that $A$ is non negative with two simple arithmetical inequalities:

\begin{claim}\label{claim:arithmetic1}
If $\ell,\  \ell', \ r,\ r', j$ are five integers such that $r'\geq r$, $\ell' \geq \ell$ and $r \geq j \geq \ell -1 $, then 
	$$ (r' - j)(r- \ell + 1) \geq  (r - j)(r' - \ell' + 1) \, .$$
\end{claim}

\begin{claim}\label{claim:arithmetic2}
If $\ell,\  \ell', \ r,\ r', j$ are five integers such that $r'\geq r$, $\ell' \geq \ell$ and $r' + 1 \geq j \geq \ell'  $, then 
	$$ (j - \ell)(r' - \ell' + 1) \geq  (j - \ell')(r - \ell + 1) \, .$$
\end{claim}
Observe that some entries $A_{i \, k}$ with $\ell \geq k \geq r'-1$ may be null if $\ell = \ell'$ or $r = r'$.
Moreover the $i$-th line of $A$ is null if both $\ell = \ell'$ and $r = r'$, in  which case agents~$i$ and $i+1$ 
	have  merged by time $t+1$ since they have the same neighbors at time $t$.
\end{proof}

\begin{prop}\label{prop:stochastic}
If  there is no cut at time~$t$, then $A_t$ is column-stochastic.
\end{prop}

\begin{proof}

If there is no cut at time $t+1$ and so at time $t$, we use the identity~(\ref{eq:perimeter}) to obtain
	\begin{equation} \label{eq:stochastic}
	\sum_{k \in [n]}  x_k^*(t) = \sum_{k \in [n]} S_k(t) \,  x_k^*(t) \, ,
	\end{equation}
	where $S_k(t)$ is the sum of $A_t$'s entries in the $k$-th column.
	
We now prove that  for every $k \in [n]$, $S_k(t) = 1$.
Again we omit $t$ in the notation as no confusion can arise. 
We consider two cases:
\begin{enumerate}
\item For each index $k \mod n$,  $ \delta(  \xb_{k -1}   , \xb_k   ) < 1$ or $\delta( \xkb   , \xb_{k + 1}   ) <1$.

We now fix an index $i$ and show that $S_{i-1} = S_i $.
For sufficiently small but non-null variations of $i$'s position on the circle,
	there is  no change in the influence graph.
Formally, in the case $ \delta(  \xb_{i-1}   , \xib   ) < 1 $, there exists $\varepsilon >0$ such that 
	for the position vector $\yb$ whose each
	entry $\yb_k$ is equal to $\xb_k $, except $\yb_i = \xib  - \varepsilon$,
	the influence graph is the same as for $\xb $.
Hence  $$ y^* = \left\{ \begin{array}{ll}
					x_k^*  - \varepsilon & \mbox{ if } k= i -1 \\
					x_{k}^*  + \varepsilon & \mbox{ if } k= i  \\
					x_k^*   & \mbox{ otherwise. } 
					\end{array} \right. \, $$
From (\ref{eq:stochastic}) and $ \varepsilon >0$, it follows that  $S_{i-1} = S_i$. 

Since $p$ is positive, then
	we conclude that for every $k \in [n]$, $S_k = 1$.
	
\item For some index $k \mod n$, $ \delta(  \xb_{k -1}   , \xb_k   ) =   \delta( \xkb   , \xb_{k + 1}   )  = 1  $.
Let us denote $$d = \min \{ \delta(\xib , \xjb ) \ | \ i \in [n],\ j \notin N_i  \} \, $$
	and  let us consider the influence graph for the position vector $\xb $ and the influence radius $\varrho = \frac{d +1}{2}$.
Since there is no cut, all agents have not merged and $d >1$.
Hence   we have $$d > \varrho > 1\, .$$
Therefore 
	$$ \delta(  \xb_{i}   , \xb_j   ) \leq  1 \Rightarrow  \delta(  \xb_{i}   , \xb_j   ) \leq \varrho \, ,$$
 and by definition of $d$, 
	$$ \delta(  \xb_{i}   , \xb_j   ) >  1 \Rightarrow  \delta(  \xb_{i}   , \xb_j   ) \geq d > \varrho \, .$$
In other words,  the influence graph is the same with the influence radius 1 and $\varrho$ when the position vector is $\xb $ .
Since for each index $k$, $ \delta(  \xb_{k -1}   , \xb_k   ) < \varrho   \   \mbox{ and }  \  \delta( \xkb   , \xb_{k + 1}   ) < \varrho$,
	we conclude  as in case~1.
\end{enumerate}	
\end{proof}
Importantly the above proposition does not hold anymore in the case of a cut in the circle.

\begin{prop}\label{prop:rooted}
If there is no cut at time~$t$, then the directed graph associated to $A_t$ is rooted.
\end{prop}

\begin{proof}
Let $H_t$ denote the directed graph associated to $A_t$.
As no confusion can arise, we omit $t$ in the notation $H_t,G_t, r_i(t), \dots$.

From  the expression~(\ref{eq:entries}) of $A_t$'s entries, we derive that there is a link $(i,i+1)$ in~$H$
	if and only if $r_{i+1} -1 \geq i +1$.
Because there is no cut, the latter inequality holds if and only if~$i$ and $i+1$ have not merged by time~$t+1$.
That proves  the following lemma.
\begin{lem}
There is a link $(i,i+1)$ in~$H$ if and only if $i$ and $i+1$ have not merged by time~$t+1$.
\end{lem}

Since there is no cut, all the agents have not merged at time~$t+1$; 
	let $i_1$ be the first $i \in [n]$ such that agents~$i$ and~$i+1$  have not  merged.
Thus  $(i_1, i_1+1)$ is an edge in~$H$.
	
Now we inductively construct a spanning tree  rooted at~$i_1$ contained in $H$.
For easier notation, we let $i_1=1$.

\begin{enumerate}
\item At the first step, we have the subtree $T_1$ over the set of nodes $\{1,2\}$  and the link $(1,2)$.
\item Suppose that at the end of the $(i-1)$-th step, the resulting directed graph  $T_{i-1}$ over the set of nodes $\{1, \dots, i \}$
	is a tree rooted  at~$1$  with the last level composed of the sole links $(i-k, i-k+1), \dots, (i-k,i)$ if $k$ denotes
	the unique integer in $\{1, \dots, i - 1\}$	such that
	$$\xb_{i-k+1}(t+1) = \dots =  \xib(t+1)\ \mbox{ and } \ \xb_{i-k}(t+1)
\neq  \xib(t+1)$$
	(cf.\ Figure~\ref{fig:tree}).

\begin{figure}
\begin{tikzpicture}[>=latex']
\begin{scope}[shift={(0,0)}]
\draw [thick,domain=250:290] plot ({8*cos(\x)}, {8*sin(\x)});

\draw [<->,domain=260:280] plot ({7.6*cos(\x)}, {7.6*sin(\x)});
\draw [|-|,domain=260:280] plot ({7.6*cos(\x)}, {7.6*sin(\x)});

\node at  ({7.3*cos(270)}, {7.3*sin(270)}) {$\leq 1$};

\node[draw,circle,fill,label=below:{$\bar{x}_{i-k}$}] at ({8*cos(260)},{8*sin(260)}) {};
\node[draw,circle,fill,label=below:{$\bar{x}_{i-k+1} = \cdots = \bar{x}_i$}] at ({8*cos(280)},{8*sin(280)}) {};
\end{scope}

\begin{scope}[shift={(7,-7.6)}]
\node [draw,circle,label=below:{$i-k$}] (r1) at (0,0) {};
\node [draw,circle,label=right:{$i-k+1$}] (c1) at (2,1) {};
\node [label=right:{$\vdots$}] (c2) at (1.7,0) {};
\node [draw,circle,label=right:{$i$}] (c3) at (2,-1) {};

\draw [->] (r1) -- (c1);
\draw [->] (r1) -- (c2);
\draw [->] (r1) -- (c3);
\draw [->] (-0.5,0) -- (r1);
\end{scope}

\begin{scope}[shift={(4,-7.6)}]
\draw [very thick,->,style={decorate, decoration={snake,
post length=1mm}}] (0,0) -- (1,0);
\end{scope}
\end{tikzpicture}
\caption{Translation of merged agents into the construction of tree $T_{i-1}$;
all positions are at time $t+1$}
\label{fig:tree}
\end{figure}
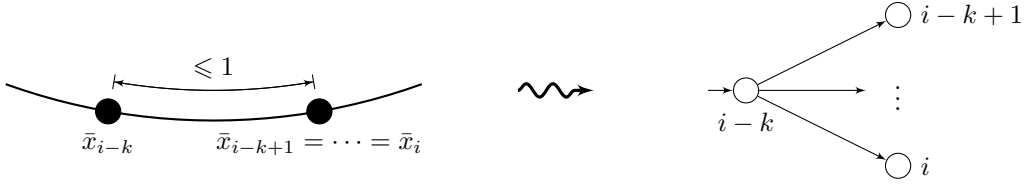

\begin{enumerate}
\item If $\xb_{i+1}(t+1) =  \xib(t+1)$, then $\xb_{i-k}(t+1) \neq  \xb_{i+1}(t+1)$ and 
	$(i-k,i+1)$  is a link in $H$.
\item Otherwise $i$ and $i+1$ have not merged by time $t+1$, and  $(i,i+1)$ is a link in $H$.
\end{enumerate}	
Then we extend the subtree $T_{i-1}$ to the set of nodes $\{1, \dots, i + 1 \}$  by adding
	the link 	$(i-k,i+1)$ or $(i,i+1)$, accordingly.
\end{enumerate}	
By construction, the resulting directed graph $T_{n-1}$ is a  tree,  rooted at~$1$, and all its links are in $H$. 
\end{proof}

As a consequence of the theorem on the backward product of line-stochastic matrices with oriented
	associated graphs (or equivalently, on the forward product of column-stochastic matrices with 
	rooted associated graphs) proved by Cao, Morse, and Anderson~\cite{CMA08a}, we derive the
	following convergence result on the sequence $\left (x^*(t) \right
)_{t\in \IN}$, taking into account the fact that matrix~$A_t$ is eventually
constant.

\begin{cor}\label{cor:conv*}
If there is never a cut, then 
the sequence $\big( x^*(t)\big)_{t\in \IN}$ is convergent and $\lVert x^*(t) - v  \lVert = O \left(\varrho^t \right)$
	for  some $\varrho\in]0,1[$ if     $v = \lim x^*(t) $.
	\end{cor}
In~\cite{CBFN14}, we proved that Corollary~\ref{cor:conv*} holds with $\varrho=1-n^{-n}$.

\section{Convergence of the  HK dynamics on the circle}\label{sec:conv}

We now put all the pieces together to show the convergence of the  HK dynamics on the circle.
\vspace{0.2cm}
\begin{thm}
An HK system on the circle converges asymptotically.
\end{thm}

\begin{proof}
If there is a cut, the dynamics on the circle is the same as on the line, in which case
	the convergence is well-known.
	
In the case no cut ever occurs, we prove a refinement of Theorem~\ref{thm:kinetic2}.
Let $t_0$ be the time at which the influence graph does not change anymore, and let $E$
	denote the set of links in the final influence graph.
By definition of the Lyapunov function $W$, for any $t$ and $t'$, $t_0\leq t\leq t'$, 
\begin{equation}\label{eq:W}
	W(t) - W(t') = \sum_{(i,j) \in E} \delta(\xib(t), \xjb(t))^2 -  \delta(\xib(t'), \xjb(t'))^2 \, .
	\end{equation}
We have $$ \delta(\xib(t), \xjb(t))  + \delta(\xib(t'), \xjb(t')) \leq 2r \leq p \, .$$
Moreover since $r<p/2$, we obtain
	\begin{equation*}\begin{split}
	 |\delta(\xib(t), \xjb(t))  - \delta(\xib(t'), \xjb(t'))| &
\leq 
	\left| \delta(\xib(t), \xb_{i+1}(t))  - \delta(\xib(t'),
\xb_{i+1}(t'))\right|  + \dots + \\ & \qquad +\left | \delta(\xb_{j-1}(t), \xjb(t))  - \delta(\xb_{j-1}(t'), \xjb(t')) \right | \\
	  & \leq  \left | x_i^*(t) -  x_i^*(t') )\right |  + \dots +\left |  x_{j-1}^*(t) -  x_{j-1}^*(t')  \right | \ .
	  \end{split} 
\end{equation*}
Then we use the above inequalities to bound each term in~(\ref{eq:W}). 

From Corollary~\ref{cor:conv*}, it follows that if there is never a cut, then 
	$$ W(t) - W(\infty) \leq C n^3 \varrho^t  \, $$
	for some positive constant $C$.
By Proposition~\ref{prop:W}, for each $i\in [n]$ 
	$$\delta(\xib(t) , \xib(t+1))^2 \leq  W(t) - W(\infty) \, . $$
Therefore $$	\delta(\xib(t) , \xib(t+1)) \leq \sqrt{C} n^{3/2} \varrho^{t/2}  \, $$
	and so $	\delta(\xib(t) , \xib(t+1)) = O \left ( \varrho^{t/2}  \right )$.
	
This establishes the convergence of each sequence $\big( \xib(t)\big)_{t\in \IN}$.

\end{proof}

\section{An alternative proof of convergence}\label{sec:v2}

In the second version of their paper~\cite{HMW15}, Hegarty et al. give an alternative proof of Corollary~\ref{cor:conv*}
	that  uses neither  the column-stochasticity of the matrix~$A_t$ nor the rootedness of its associated graph
	(Propositions~\ref{prop:stochastic} and~\ref{prop:rooted}, respectively).
In fact, their new argument to prove the exponential convergence of $\big( x^*(t)\big)_{t\in \IN}$ also
	works for the position vectors as we will show below.
Thus the resulting proof of the convergence of the HK dynamics on the circle directly follows from the finiteness of 
	the quadratic kinetic energy $K_2$ and the eventual stability of influence graphs.

\vspace{0.2cm}
For simplicity, let us denote by $\dot{x}(t)$ the vector whose $i$-th component is $\vect(\overline{x}_{i}(t) , \xb_{i}(t + 1))$, i.e.,
	$$ \dot{x}_i(t)	=    \frac{1}{|N_{i}(t)|} \sum_{k\in N_{i}(t)} \vect(\overline{x}_{i}(t) , \xkb(t))    \, .$$
Under the condition $r<p/6$, for each $i$'s neighbor $k$ at time $t$ we have
	$$  \vect(\overline{x}_{i}(t) , \xkb(t) )  =  \vect(\overline{x}_{i}(t) , \xib(t -1) )   +  \vect(\overline{x}_{i}(t -1) , \xkb(t -1))   +  \vect(\overline{x}_{k}(t -1 ) , \xkb(t))   \, .$$
Therefore, $$   \dot{x}_i(t) =  - \dot{x}_i(t - 1)    +  \frac{1}{|N_{i}(t)|} \sum_{k\in N_{i}(t)}  \vect(\overline{x}_{i}(t -1) , \xkb(t-1 ))  
                                                                                         +    \frac{1}{|N_{i}(t)|} \sum_{k\in N_{i}(t)}  \vect(\overline{x}_{k}(t -1) , \xkb(t ))  \, .$$
From time $t_0$,   the influence graph is constant and the neighborhood of  each agent does not vary anymore.
For $t > t_0$, it follows that $   \frac{1}{|N_{i}(t)|} \sum_{k\in N_{i}(t)}  \vect(\overline{x}_{i}(t -1) , \xkb(t-1 ))  =  
	  \dot{x}_i(t-1) $, and thus
	  \begin{equation}\label{eq:dotx}
	  \dot{x}(t)  = B  \, \dot{x}(t-1) \,   \end{equation}
	  where $B$ denotes the line-stochastic matrix whose associated graph is the influence graph at time $t_0$ and 
	  with positive entries in the $i$-th line equal to $1/|N_i(t_0)|$.
	  
By Theorem~\ref{thm:kinetic2}, the sequence of vectors 	 $\dot{x}(t)$ tends to the zero vector.
Using~(\ref{eq:dotx}) and  the Jordan normal form of $B$, it follows that 	each component of  $\dot{x}(t)$
	is in  $ O\left( e^{-ct} \right)$, which proves the finiteness of $K_1$, and  thus the convergence of the HK dynamics on the circle.
	  \vspace{0.7cm}

\bibliographystyle{plain}

\bibliography{../agents}

\end{document}